\documentclass[12pt]{article}

\usepackage{epsfig}

\usepackage{amssymb}
\usepackage{amsfonts}

\def\be{\begin{equation}}
\def\ee{\end{equation}}
\def\ba{\begin{array}{c}}
\def\ea{\end{array}}

\newcommand{\pkt}{\!\!\succ\,\,}
\newcommand{\kt}{\rangle}
\newcommand{\br}{\langle}

\newtheorem{thm}{Theorem}

\newtheorem{lemma}[thm]{Lemma}

\newenvironment{proof}{\noindent
 {\bf Proof.}}{\hfill$\square$\vspace{3mm}\endtrivlist}

\begin{document}

%

\begin{center}

{\Large

Perturbation theory near degenerate exceptional points

 }

\vspace{0.8cm}

  {\bf Miloslav Znojil}

\vspace{0.2cm}

\vspace{0.2cm}

\vspace{1mm} The Czech Academy of Sciences, Nuclear Physics
Institute,

Hlavn\'{\i} 130, 25068 \v{R}e\v{z}, Czech Republic,

\vspace{0.2cm}

%
%
%
%
%
%
%
%

 and

\vspace{0.2cm}

Department of Physics, Faculty of Science, University of Hradec
Kr\'{a}lov\'{e},

Rokitansk\'{e}ho 62, 50003 Hradec Kr\'{a}lov\'{e},
 Czech Republic

%
%

\vspace{0.2cm}

 {e-mail: znojil@ujf.cas.cz}

\vspace{0.3cm}


\end{center}

\newpage

\section*{Abstract}


In an overall framework of 
quantum mechanics of unitary systems
a rather sophisticated new version of perturbation
theory is developed and described.
The 
motivation of such an extension of the list of the currently available
perturbation-approximation recipes was four-fold:
(1) its need results from 
the quick growth of interest in 
quantum systems exhibiting 
parity-time symmetry (${\cal PT}-$symmetry)
and its generalizations;
(2) in the context of physics, 
the necessity of a thorough update of perturbation theory
became clear immediately after the identification of
a class of quantum phase transitions with the non-Hermitian 
spectral degeneracies 
at the Kato's
exceptional points (EP); 
(3) in the dedicated
literature, the EPs are only being
studied in the special
scenarios characterized by the spectral
geometric multiplicity $L$ equal to one;
(4) apparently, one of the decisive reasons may be seen in the 
complicated nature of mathematics behind the 
$L\geq 2$ constructions.
In our present paper we show how to overcome the latter,
purely technical obstacle. The temporarily
forgotten class of the $L>1$ models
is shown accessible to a feasible perturbation-approximation
analysis. In particular, an emergence of a counterintuitive
connection between the value of $L$, the structure of 
the matrix elements of perturbations, and
the possible loss of the stability and unitarity of the processes 
of the unfolding
of the singularities is given a detailed explanation.

\subsection*{Keywords}
.

non-Hermitian quantum dynamics;

unitary vicinity of exceptional points;

degenerate perturbation theory;

Hilbert-space geometry near EPs;

 \newpage

\section{Introduction}

The Bender's and Boettcher's \cite{BB} idea of replacement of
Hermiticity $H = H^\dagger$ by parity-time symmetry (${\cal
PT}-$symmetry) $H{\cal PT}={\cal PT}H$ of a Hamiltonian responsible
for unitary evolution opened, after an appropriate mathematical
completion \cite{Carl,ali,book,Carlb} of the theory, a way towards the
building of quantum models exhibiting non-Hermitian degeneracies
\cite{Berry} {\it alias\,} exceptional points (EPs, \cite{Kato}).
For a fairly realistic illustrative example
of possible applications opening multiple new horizons
in phenomenology we could recall,
e.g., the well known phenomenon of Bose-Einstein condensation
is a schematic simplification described by
the non-Hermitian but ${\cal PT}-$symmetric
three-parametric Hamiltonian
 \begin{equation}
 \label{Ham1}
  H^{(GGKN)} (\gamma,v,c)= -{\rm i} \gamma
  \left(a_1^{\dagger}a_1 - a_2^{\dagger}a_2\right) +
  v\left(a_1^{\dagger}a_2 + a_2^{\dagger}a_1\right) +
  \frac{c}{2}
  \left( a_1^{\dagger}a_1 - a_2^{\dagger}a_2\right)^2\,.
 \end{equation}
This Hamiltonian represents an interesting analytic-continuation
modification of the
conventional Hermitian
Bose-Hubbard Hamiltonian \cite{[38],[37],zaUwem}.
In this form the model
was recently paid detailed attention in
Ref.~\cite{Uwe}.
A consequent application of multiple, often fairly sophisticated forms of
perturbation theory has been shown there to lead to surprising results.
In particular,
the behavior of the bound and resonant states of the
system was found to lead to the new and unexpected
phenomena
in the dynamical regime characterized by the small
coupling constant $c$.
The authors of Ref.~\cite{Uwe} emphasized that new physics may be
expected to emerge precisely
in the vicinity of the EP-related dynamical singularities.

These phenomena (simulating not necessarily just the
Bose-Einstein condensation of course) were analyzed, in \cite{Uwe}, using
several {\it ad hoc}, not entirely standard perturbation techniques.
The
role of an unperturbed Hamiltonian $H_0$ was assigned, typically, to
the extreme EP limits of $H$. Unfortunately, only too often
the perturbed energies appeared to be complex as a consequence.
In other words, the systems
exhibiting ${\cal
PT}-$symmetry
seemed to
favor the spontaneous breakdown of this symmetry near EPs.

In the light of similar results one immediately has to ask the
question whether such a ``wild behavior''  of the EP-related quantum systems
is generic.
Indeed, an affirmative answer is often encountered in
the studies by mathematicians (see, e.g., \cite{Viola}). A hidden reason is that
they usually tacitly keep in mind just the
``effective theory'' and/or the so called
``open quantum system'' dynamical scenario \cite{Nimrod}.

In the more restrictive context of the unitary quantum mechanics
the situation is different: the ``wild behavior''
of systems is usually not generic there
(cf. also the recent explanatory commentary
on the sources of possible misunderstandings in \cite{Ruzicka}).
Several
non-numerical illustrative models may be found in \cite{BHAO}
where, typically, the ${\cal PT}-$symmetric
Bose-Hubbard model of Eq.~(\ref{Ham1})
(which behaves as unstable near its
EP singularities \cite{Uwe}))
has been replaced by its ``softly perturbed'' alternative
in which, in an arbitrarily small vicinity of
its EP singularity,  the system remains stable and unitary
under admissible perturbations.

The physics of stability covered by
paper \cite{BHAO} can be perceived as
one of the main sources of inspiration
of our present study. We intend
to replace here the very specific
model of Eq.~(\ref{Ham1})
(in which the geometric multiplicity $L$
of all of its
EP-related degeneracies
was always equal to one)
by a broader class of quantum systems.
In a way motivated by the idea of a highly desirable extension
of the
currently available menu of
the tractable and eligible dynamical scenarios
beyond their $L=1$ subclass,
we will turn attention here
to the
EP-related degeneracies
of the larger,
nontrivial geometric multiplicities $L\geq 2$.
We will reveal that such a study opens new horizons
not only in phenomenology
(where the influence of perturbations becomes
strongly dependent on the
detailed structure of the non-Hermitian degeneracy)
but also in mathematics (where a rich menu of
physical consequences will be shown reflected
by an unexpected
adaptability of the geometry of the Hilbert space
to the detailed structure of the perturbation).


The presentation of our results will be organized as follows. 
Firstly, in section~\ref{seaone} we will recall a typical
quantum system (viz., a version of the non-Hermitian Bose-Hubbard 
multi-bosonic model) 
in which the EP degeneracies play a decisive phenomenological role.
We will explain that although the model itself only exhibits the 
maximal-order $L=1$
EP degeneracies, such a option represents, from the purely formal point of view, 
just one of the eligible dynamical scenarios. In Appendix A 
a full classification of the EPs is presented therefore, showing, i.a., that
the number of the ``anomalous'' EPs 
of our present interest with $L \geq 2$ exhibits an almost exponential growth
at the larger matrix dimension $N$.

The goals of our considerations are subsequently explained in section
\ref{seatwo}. For the sake of brevity, we just pick up
the first nontrivial case with $L=2$, and we emphasize that 
even in such a case the basic features of 
an appropriate adaptation of perturbation theory may be explained,
exhibiting also, not quite expectedly,
a survival of the fairly user-friendly mathematical structure.

In order to make our message self-contained,
the known form of the EP-related perturbation formalism 
restricted to $L=1$ is reviewed
in Appendix B. On this background, 
in a way based on a not too dissimilar constructive
strategy, our present main $L \geq 2$ results are then presented 
and described in section
\ref{seathree}. 
We emphasize there the existence of
the phenomenological as well as mathematical subtleties of the large$-L$ models.
We show that 
in our generalized, degenerate-perturbation-theory formalism
a key role is played by
an interplay between its formal mathematical background
(viz, the non-Hermiticity of the Hamiltonians)
and its phenomenological aspects
(typically, the knowledge of $L$ must be complemented by 
an explicit knowledge of the partitioning of 
Schr\"{o}dinger equation).

In section
\ref{seafour},
all of these aspects of the $L>1$ perturbation theory are
summarized and illustrated by a detailed description
of the characteristic,
not always expected features of the leading-order approximations. 
Several related applicability aspects of our present 
degenerate-perturbation-theory
formalism are finally discussed in section
\ref{seafive} and in two Appendices.
We point out there that
some of the features of the $L>1$ theory 
(e.g., a qualitative, fairly counterintuitive clarification
of the concept of the smallness of perturbations)
may be 
treated as not too different from their $L=1$ predecessors.
At the same time, a wealth of new formal challenges 
is emphasized to emerge, in particular, 
in the analyses of the role of perturbations in the 
specific quantum systems which are required unitary.

\section{Exceptional points\label{seaone}}

\subsection{Bose-Hubbard model and exceptional points
of geometric multiplicity one, $L=1$}

Illustrative ${\cal PT}-$symmetric Bose-Hubbard Hamiltonian operator
(\ref{Ham1}) commutes with the number operator
 \be \label{Num1}
 \widehat{N}=a_1^{\dagger}a_1 + a_2^{\dagger}a_2\,.
 \ee
This means that the number of bosons $N$ is conserved so that after
its choice the Hamiltonian may be represented by a
finite-dimensional non-Hermitian $K$ by $K$ matrix $ H^{(GGKN)}
(K,\gamma,v,c)$ with $K = N+1$ (see its explicit construction in
\cite{Uwe}). Once we fix the units (such that $v=1$) and once we set
$c= 0$ (preserving, for the sake of simplicity, just the first two
components of the Hamiltonian), the resulting one-parametric family
of Hamiltonian matrices $ H^{(GGKN)} (K,\gamma)$ can be assigned the
closed-form energy spectra
 \be
 E_{n}^{(GGKN)}(K,\gamma)
 =(1-\gamma^2)^{1/2}\,(1-K+2n)\,,
 \ \ \ \ \ n=0,1, \ldots,  K-1\,.
 \label{spektrade}
 \ee
These energies remain real and non-degenerate
(i.e., observable)
if and only if $\gamma^2<1$.

In {\it loc. cit.} it has also been proved that the two
interval-boundary values of $\gamma = \pm 1$ are, in the terminology
of the Kato's mathematical monograph \cite{Kato}, exceptional points
(EPs, $\gamma^{(EP)}_+=1$ and $\gamma^{(EP)}_-=-1$). More precisely,
one should speak about the very special EPs of maximal order (i.e.,
of order $K$, abbreviated as EPK). The latter observation may be
given a more general, model-independent linear-algebraic background
via relation
 \be
 H^{(K)}(\gamma^{(EPK)})\, Q^{(EPK)}
 = Q^{(EPK)}\,{J}^{(K)}(\eta)\,
 \label{fealt}
 \ee
where
 \be
 \eta=\lim_{\gamma \to \gamma^{(EPK)}}
  E_{n}(K,\gamma)\,,
  \ \ \ \ n = 0, 1, \ldots, K-1
 \label{conflue}
 \ee
is the limiting degenerate energy. Relation (\ref{fealt}) contains
the so called transition matrix $Q^{(EPK)}$ and the non-diagonal,
EPK-related canonical-representation Jordan-block matrix
 \be
 J^{(K)}
 \left (\eta\right )=
 \left (
 \begin{array}{ccccc}
 \eta&1&0&\ldots&0\\
 0&\eta&1&\ddots&\vdots\\
 0&0&\eta&\ddots&0\\
 \vdots&\ddots&\ddots&\ddots&1\\
 0&\ldots&0&0&\eta
 \ea
 \right )
 \,.
  \label{JBK}
 \ee
As a certain limiting analogue of the conventional set of eigenvectors
the transition matrix is obtainable via the
solution of the EPK-related analogue (\ref{fealt})
of conventional Schr\"{o}dinger equation.
For our illustrative example
$ H^{(GGKN)} (K,\gamma)$, in particular, all of the
$K-$ and $\gamma^{(EPK)}-$dependent explicit, closed
forms of solutions ${Q^{(EPK)}}$
remain non-numerical and may be found constructed in
dedicated paper~\cite{BHAO}.

\subsection{Generic non-Hermitian degeneracies with geometric multiplicities
larger than one, $L>1$}

In the common model-building scenarios the ${\gamma}-$dependence of
Hamiltonians $H^{(K)}({\gamma})$ is analytic. Under this assumption
the exceptional points of maximal order as discussed in preceding
subsection represent just one of several possible realizations of a
non-Hermitian degeneracy (NHD) with its characteristic EP-related
confluence of eigenvalues (\ref{conflue}). Besides the maximal,
EPK-related complete confluence of eigenvectors as described in
preceding subsection we may encounter, in general, multiple other,
incomplete confluences of eigenvectors
 \be
  \lim_{{\gamma} \to {\gamma}^{(NHD)}}\,
  |{{\psi}}_{m_{k[j]}}({\gamma})\kt = |\chi_j\kt\,,
  \ \ \ \ \ \
 k[j]=1,2,\ldots, M_j\,,\ \ \ \ M_j\geq 2\,,
 \ \ \ \ \ j=1,2,\ldots,L\,.
 \label{Mde}
 \ee
Here we have $M_1+M_2+\ldots + M_L=K$ where $L$ is called the
geometric multiplicity of the EP degeneracy \cite{Kato}
(see also Appendix A for more details). 

Every EP instant ${\gamma}^{(NHD)}$ may be characterized not only by
the overall Hilbert-space dimension $K$ and by the number $L\geq 1$
of linearly independent $\eta-$related states $|\chi_j\kt$ of
Eq.~(\ref{Mde}) but also by a suitably ordered $L-$plet of the
related subspace dimensions $M_j$. Thus, in the present $L\geq 2$
extension of preceding subsection the fully non-diagonal Jordan
block of Eq.~(\ref{JBK}) must be replaced by the more general
block-diagonal canonical representation of the Hamiltonian,
 \be
  {\cal J}^{(K)}(\eta)
 =\left (
 \begin{array}{cccc}
 J^{(M_1)}\left ( \eta\right )&0&\ldots&0\\
 0&J^{(M_2)}\left ( \eta\right )&\ddots&\vdots\\
 \vdots&\ddots&\ddots&0\\
 0&\ldots&0&J^{(M_L)}\left ( \eta\right )
 \ea
 \right )
 = \bigoplus_{j=1}^L\,{ J}^{(M_j)}(\eta)
 \,.
   \label{geneve}
 \ee
In parallel, EPK relation (\ref{fealt}) must be replaced by its
generalization
 \be
 H^{(K)}(\gamma^{(NHD)})\, Q^{(NHD)}
 = Q^{(NHD)}\,{\cal J}^{(K)}(\eta)\,.
 \label{ufealt}
 \ee
Naturally, the direct-sum structure of ${\cal J}^{(K)}(\eta)$
becomes reflected by a partitioned-matrix structure of transition
matrices $Q^{(NHD)}$ which are, in general, not block-diagonal of
course.
%
%


\section{Unitary processes of collapse at $L=2$\label{seatwo}}

{\it A priori\,} one may expect that the existence of anisotropy of
the Hilbert space as realized, in Eq.~(\ref{216}) at $L=1$, by an
elementary rescaling $B(\lambda)$ of the basis will also exist at
any larger geometric multiplicity $L> 1$, i.e. in the unitary
quantum systems with the more complicated structure of the NHD
limiting {\it alias\,} quantum phase transition.

\subsection{Quantum physics behind ``degenerate degeneracies'' with $L=2$}

Hypothetically, the unitary evolution of any ${\cal PT}-$symmetric
quantum system moving towards a hiddenly Hermitian EP degeneracy
with geometric multiplicity two can be perceived as generated by a
suitable diagonalizable Hamiltonian $H^{(K)}({\gamma})$ with real
spectrum \cite{ali}. What is only necessary is that its (perhaps,
properly renumbered) eigenvectors $|{{\Phi}}_{j}({\gamma})\kt$ obey
the $L=2$ EP-degeneracy rule
 \be
  \lim_{{\gamma} \to {\gamma}^{(EP)}}|{{\Phi}}_{m}({\gamma})\kt
  = |\chi_a\kt\,,
  \ \ \ \ \ \
 m=0,1,\ldots, M-1\,,
 \label{Mdegewa}
 \ee
 \be
  \lim_{{\gamma} \to {\gamma}^{(EP)}}|{{\Phi}}_{M+n}({\gamma})\kt
  = |\chi_b\kt\,,
  \ \ \ \ \ \
 n=0,1,\ldots, N-1\,.
 \label{uMdegewa}
 \ee
The two limiting eigenvectors $|\chi_a\kt $ and $|\chi_b\kt$ are, by
our assumption, linearly independent so that partitioned matrix
(\ref{geneve}), i.e., at $L=2$ and $\eta=0$, matrix
 \be
 {\cal J}^{(M \bigoplus N)}(0)
 =\left[ \begin {array}{cc} J^{(M)}({0})&0
 \\{}0&J^{(N)}({0})
 \end {array}
 \right]\,
 \label{ursset}
 \ee
will represent the canonical form of our Hamiltonian in the EP
limit. The corresponding transition matrix $Q^{(M \bigoplus N)}$ can
be then obtained by the solution of the limiting version
 \be
 H^{(K)}(\gamma^{(EP)})\, Q^{(M \bigoplus N)}
 = Q^{(M \bigoplus N)}\,{\cal J}^{(M \bigoplus N)}(0)\,
 \label{ealt}
 \ee
of the initial Schr\"{o}dinger equation. A few exactly solvable
samples of the latter $L=2$ degeneracy process ${\gamma} \to
{\gamma}^{(EP)}$ can be found, e.g., in our recent paper~\cite{NPB}.

In all of the similar dynamical scenarios one can always find
parallels with their simpler $L=1$ predecessors. In particular, a
return to the situation before the collapse as sampled, at $L=1$, by
Eq.~(\ref{berefealt}) above, can be also given the following
analogous form
 \be
 \left [Q^{((M \bigoplus N)}\right ]^{-1}\,
 H^{(K)}(\gamma)\, Q^{(M \bigoplus N)}=
 {J}^{(M \bigoplus N)}(0) +  \lambda\,V^{(K)}(\gamma)\,,
 \ \ \ \ \gamma \neq \gamma^{(EP)}
  \label{derefealt}
 \ee
of a ``unitarity-compatible'' perturbation-theoretic
reinterpretation in which the perturbation
$\lambda\,V^{(K)}(\gamma)$ is fully determined by the input matrix
Hamiltonian $H^{(K)}(\gamma)$.

\subsection{Unfoldings of degeneracies under
random perturbations at $L=2$}

In a close parallel to the $L=1$ Schr\"{o}dinger's bound-state
problem (\ref{perL1}) let us now start the study of its $L>1$
generalizations by considering the first nontrivial choice of
degeneracy with the geometric multiplicity $L=2$. In our present
notation the corresponding Schr\"{o}dinger equation then reads
   \be
 \left [
 {\cal J}^{(M \bigoplus N)}(0) + \lambda\,V^{(K)}
 \right ]
 \,|{\Psi}\kt=\varepsilon\,|{\Psi}\kt\,.
 \label{epmod}
 \ee
Without any loss of generality we set again $\eta=0$. After such a
choice all of the eigenvalues $\varepsilon=\varepsilon(\lambda)$
will remain small, and they will vanish in the formal
unperturbed-system limit $\lambda \to 0$.

In a way paralleling Eq.~(\ref{ursset}), any given matrix of
perturbations has to be partitioned as well,
 \be
 V_{}^{(K)}=\left[ \begin {array}{cc}
  V^{(M,M)}_{}&V^{(M,N)}_{}
 \\{}V^{(N,M)}_{}&V^{(N,N)}_{}
 \end {array}
 \right]\,.
 \label{pema}
 \ee
With the four submatrices having dimensions indicated by the
superscripts we shall assume, at the beginning at least, that all of
the individual matrix elements of the perturbation matrix $V^{(K)}$
remain bounded at small $\lambda$. This means that the Hamiltonian
is dominated by its unperturbed part  (\ref{ursset}) so that also
the whole perturbed $L=2$ Schr\"{o}dinger equation (\ref{epmod}) has
to be partitioned. In order to simplify the notation we shall write
 \be
 |\Psi\kt=
 \left (
 \ba
 |\psi^{(a)}\pkt\\
 |\psi^{(b)}\pkt
 \ea
 \right )\,,
 \ \ \ \
 |\psi^{(a)}\pkt=
 \left (
 \ba
 \psi_{1}^{(a)}\\
 \psi_{2}^{(a)}\\
 \vdots\\
 \psi_{M}^{(a)}
 \ea
 \right )\,,\ \ \ \
 |\psi^{(b)}\pkt=
 \left (
 \ba
 \psi_{1}^{(b)}\\
 \psi_{2}^{(b)}\\
 \vdots\\
 \psi_{N}^{(b)}
 \ea
 \right )
 \,
 \label{rajt}
 \ee
using the curly-ket symbols for subvectors. This will enable us to
proceed in a partial parallel with the widely studied 
non-degenerate-EP cases where one has
$L=1$ (see Ref.~\cite{admissible,corridors} 
or Appendix B below for a compact review).


\section{Perturbation theory at $L=2$\label{seathree}}

\subsection{The recent change of the unitary-evolution paradigm}

From the historical point of view the use of EPs in physics has not
been immediate. Only during the last cca 20 years one notices a
perceivable increase of the relevance of the concept in various
branches of theoretical as well as experimental physics. Various
innovative EP applications emerged ranging from the analyses of
resonances in classical mechanics \cite{Mailybaev} and of the so
called non-Hermitian degeneracies in classical optics \cite{Berry}
up to the studies of the wealth of phenomena in quantum physics of
open quantum systems \cite{Nimrod,Nimrodb} or, last but not least, even of
the closed, stable and unitarily evolving quantum systems
\cite{BB,Carl}.

All of these developments contributed to the motivation of our
present study. For the sake of definiteness we restricted our
attention to the framework of quantum physics, unitary or
non-unitary. In this setting the traditional role of the EPs
${\gamma}^{(EP)}$ has always been twofold. Firstly, in the context
of mathematics, the conventional analyticity assumptions about
Hamiltonians
 \be
 H({\gamma})=H({\gamma}_0)+({\gamma}-{\gamma}_0)\,H^{(1)}
 +({\gamma}-{\gamma}_0)^2\,H^{(2)}
 + \ldots\,,
 \ee
and  the conventional power-series ansatz for energies
 \be
  E_n({\gamma})=E_n({\gamma}_0)+({\gamma}-{\gamma}_0)\,E_n^{(1)}
 +({\gamma}-{\gamma}_0)^2\,E_n^{(2)}
 + \ldots\,
 \label{seres}
 \ee
(etc) gave birth to the so called Rayleigh-Schr\"{o}dinger
perturbation-expansion constructions of the Schr\"{o}dinger-equation
solutions. It has been revealed that the radius of convergence $R$
of these perturbation-series solutions is determined by the position
of the nearest EP in the complex plane of the parameter, $R=\min
|{\gamma}_0-{\gamma}^{(EP)}|$ \cite{Kato}.

In the other, direct applications of EPs, the localization of
singularities ${\gamma}^{(EP)}$ only played an important traditional
role in non-unitary, open quantum systems \cite{Nimrod}. An
explanation is easy:  for any self-adjoint Hamiltonian $H({\gamma})$
characterizing a closed quantum system the necessary reality of the
parameter (${\rm Im\,}{\gamma} = 0$) cannot be made compatible with
the fact that {\em all\,} of the values of the eligible (i.e., not
accumulation-point) EP parameters are complex, ${\rm
Im\,}{\gamma}^{(EP)} \neq 0$.

The traditional paradigm has only been changed recently, after
Bender with Boettcher \cite{BB} managed to turn attention of
physicists' community to the existence of a broad class of
Hamiltonians $H({\gamma})$ which happen to be non-Hermitian but
parity-time symmetric (${\cal PT}-$symmetric) in ${\cal K}$. One of
the characteristic mathematical features of these Hamiltonians is
that in spite of their non-Hermiticity, their {\it whole\,} spectrum
$\{E_n\}$ may remain strictly real in a suitable {\em real\,}
interval ${\cal D}$ of the unitarity-compatible parameters
${\gamma}$ (see, e.g., monograph \cite{book} for more details).

\subsection{Rearrangement of Schr\"{o}dinger equation}

With the two subscripts $j_{(a)}$ and $j_{(a)}$ running, in the
curly-ket subvectors in (\ref{rajt}), from $1$ to  $N_{(a)}=M$ and
$N_{(b)}=N$, respectively, we will now only partially fix the norm
by setting $\psi^{(a)}_{1}=\omega_{(a)}$ and
$\psi^{(b)}_{1}=\omega_{(b)}$ or, in a self-explanatory shorthand,
$\psi^{(a,b)}_{1}=\omega_{(a,b)}$. Next, a parallel to the $L=1$
redefinition~(\ref{jtarov}) of wave functions will be found in its
$L=2$ extension
 \be
 |\vec{y}^{(a,b)}\pkt=
 \left (
 \ba
 y^{(a,b)}_{1}\\
 y^{(a,b)}_{2}\\
 \vdots\\
 y^{(a,b)}_{N_{(a,b)} -1}\\
 y^{(a,b)}_{N_{(a,b)} }
 \ea
 \right )=
 \left (
 \ba
 \psi^{(a,b)}_{2}\\
 \psi^{(a,b)}_{3}\\
 \vdots\\
 \psi^{(a,b)}_{N_{(a,b)} }\\
 \Omega_{(a,b)}
 \ea
 \right )\,
 \label{rearr}
  \ee
where the two new, temporarily variable elements $\Omega_{(a)}$ and
$\Omega_{(b)}$ will have to be determined later. In terms of the
four auxiliary symbols
 \be
 |e^{(a,b)}\pkt=\left (
 \ba
 1\\
 0\\
 \vdots\\
 0
 \ea
 \right )
 \,,
 \ \ \ \ \
 \Pi^{(a,b)}=
 \left[ \begin {array}{ccccc} 0&0&0&\ldots&0
 \\{}1&0&0&\ddots&\vdots
 \\{}0&\ddots&\ddots&\ddots&0
 \\{}\vdots&\ddots&1&0&0
 \\{}0&\ldots&0&1&0
  \end {array}
 \right]\,
 \label{arab}
 \ee
of dimensions $N_{(a)}=M$ and $N_{(b)}=N$ we will further decompose
 $$
 |\psi^{(a,b)}\pkt=|e^{(a,b)}\pkt \!\omega_{(a,b)}
 +\Pi^{(a,b)}\,|\vec{y}^{(a,b)}\pkt\,.
 $$
In the next step we introduce the $L=2$ analogue of the
unpartitioned $L=1$ vector (\ref{postaru}),
 \be
 |r\kt=
 \left (
 \ba
 |r^{(a)}\pkt\\
 |r^{(b)}\pkt
 \ea
 \right )\,,
 \ \ \ \
 |r^{(a,b)}\pkt=
 \left (
 \ba
 r_{1}^{(a,b)}\\
 r_{2}^{(a,b)}\\
 \vdots\\
 r_{N_{(a,b)}}^{(a,b)}
 \ea
 \right )\,
 \ee
with components
 \be
 r_{i}^{(a,b)}=
 \varepsilon\,\omega_{(a,b)}\,\delta_{i,1}
 - \lambda\,
  V^{(N_{(a,b)},M)}_{i,1}\,\omega_{(a)}
 - \lambda\,
  V^{(N_{(a,b)},N)}_{i,1}\,\omega_{(b)}=
 r_{i}^{(a,b)}(\varepsilon,\vec{\omega})\,.
  \ee
Treating, temporarily, the two not yet specified quantities
$\Omega_{(a)}$ and $\Omega_{(b)}$ as adjustable
matrix-regularization parameters, and replacing the $L=1$ auxiliary
matrix (\ref{9}) by its partitioned $L=2$ counterpart
 $$
 {\cal A}(\varepsilon)=\left[ \begin {array}{cc}
  A(M,\varepsilon)&0
 \\{}0&A(N,\varepsilon)
 \end {array}
 \right]\,
 $$
we are just left with the problem of finding a suitable $L=2$
analogue of relation (\ref{ustahromu}).

A key to the resolution of the puzzle is found in Eq.~(\ref{arab})
and in its partitioned direct-sum extension
 $$
  {\Pi}^{(M\bigoplus N)}=\left[ \begin {array}{cc} \Pi^{(a)}&0
 \\{}0&\Pi^{(b)}
 \end {array}
 \right]\,.
 $$
Using this symbol we can now rewrite our homogeneous Schr\"{o}dinger
equation (\ref{epmod}) in the inhomogeneous matrix-inversion
representation
 \be
 \left (
 {\cal A}^{-1}(\varepsilon)
  + \lambda\,V^{(K)}\,
  \Pi^{(M\bigoplus N)}
  \right )\,\left (
 \ba
 |\vec{y}^{(a)}\pkt\\
 |\vec{y}^{(b)}\pkt
 \ea
 \right )\,
 =\left (
 \ba
 |r^{(a)}\pkt\\
 |r^{(b)}\pkt
 \ea
 \right )\,
 \label{[34]}
 \ee
or, equivalently,
 \be
 \left (
 I
  + \lambda\,{\cal A}(\varepsilon)\,V^{(K)}\,
  \Pi^{(M\bigoplus N)}
  \right )
 \,
 |\vec{y}\kt
 =
 {\cal A}\,(\varepsilon)|r\kt
  \,.
 \label{[34invr]}
 \ee
Once we drop the redundant superscripts, and once we add the
relevant parameter-dependences in (\ref{[34invr]}) we obtain
relation
 \be
 \left (
 I
  +\lambda\,{\cal A}(\varepsilon)\,V\,
  \Pi
  \right )\,|{\vec{y}}\kt
 =
 {\cal A}(\varepsilon)\,|{r(\lambda,\varepsilon,\vec{\omega})}\kt\,.
 \label{[34b]}
 \ee
This is our ultimate, iteration-friendly exact form of our perturbed
Schr\"{o}dinger equation.

\subsection{Solutions}

Equation (\ref{[34b]}) yields the ket-vector part of the solution in
closed form,
 \be
 |{\vec{y}}\kt
 =
 \left (
 I
  + \lambda\,{\cal A}(\varepsilon)\,V\,
  \Pi
  \right )^{-1}\,
 {\cal A}(\varepsilon)\,|{r(\lambda,\varepsilon,\vec{\omega})}\kt=
 |{\vec{y}}^{(solution)}(\lambda,\varepsilon,\vec{\omega})\kt
 \,.
 \label{[34c]}
 \ee
In the small-perturbation regime the latter formula may be given the
conventional Taylor-series form with
 $$\,
 |{\vec{y}}^{(solution)}(\lambda,\varepsilon,\vec{\omega})\kt=
  {\cal A}(\varepsilon)\,|{r}(\lambda,\varepsilon,\vec{\omega})\kt
 - \lambda\,{\cal A}(\varepsilon)\,V\,
  \Pi\,{\cal A}(\varepsilon)\,|{r(\lambda,\varepsilon,\vec{\omega})}\kt+
  $$
 \be
 +\lambda^2\,
   {\cal A}(\varepsilon)\,V\,
  \Pi\,{\cal A}(\varepsilon)\,V\,
  \Pi\,{\cal A}(\varepsilon)\,|{r(\lambda,\varepsilon,\vec{\omega})}\kt-\ldots
 \,.
 \label{[36]}
 \ee
Naturally, the construction is not yet finished because what is
missing is the guarantee of equivalence between the eigenvalue
problem (\ref{epmod}) and its matrix-inversion reformulation
(\ref{[34]}) containing two redundant parameters. This is the last
obstacle, easily circumvented by our setting, in solution
(\ref{[36]}), both of the redundant upper-case constants
$\Omega_{(a)}=y_{M}^{(a)}$ and $\Omega_{(b)}=y_{N}^{(b)}$ equal to
zero. In the light of explicit formula (\ref{[36]}), this
requirement is equivalent to the pair of relations
 \be
  {{y}}^{(solution)}_M(\lambda,\varepsilon,\omega_{(a)},
  \omega_{(b)})=0\,,
  \ \ \ \ \
  {{y}}^{(solution)}_{M+N}(\lambda,\varepsilon,\omega_{(a)},
  \omega_{(b)})=0
    \,.
 \label{firovni}
 \ee
Both of the left-hand-side functions of the three unknown quantities
$\varepsilon$, $\omega_{(a)}$ and $\omega_{(b)}$ are available, due
to formula (\ref{[36]}), in closed form. Although both of them can
vary with the three independent unknowns (i.e., with
$\varepsilon,\omega_{(a)}$ and $\omega_{(b)}$), one of these
variables merely plays the role of an optional normalization
constant so that we may set, say, $\omega_{(a)}^2+\omega_{(b)}^2=1$.
Thus, the implicit version of the perturbation-expansion
construction of bound states is completed.


\section{Schr\"{o}dinger equation in leading-order approximation\label{seafour}}

Two coupled equations (\ref{firovni}) determine the bound state. In
the spirit of perturbation theory one may expect that the
perturbations happen to be, in some sense, small. At the same time,
even the analysis of the comparatively elementary $L=1$ secular
equation (\ref{krutadef}) determining the single free variable
(viz., the energy) led to the necessity of a strongly
counterintuitive scaling (\ref{216}) reflecting, near the extreme EP
boundary of unitarity, the strong anisotropy of the geometry of the
physical Hilbert space. Naturally, at least comparable complications
have to be expected to be encountered during the analysis of the
more complicated set of two coupled equations (\ref{firovni})
representing the secular equation in its exact $L=2$ version,
constructed as particularly suitable for systematic approximations.

\subsection{Generic case: perturbations without vanishing elements}

Even the most drastic truncation of the formal power series
(\ref{[36]}) yields already a nontrivial ket vector
 \be
 |{\vec{y}}^{(solution)}(\lambda,\varepsilon,\vec{\omega})\kt=
  {\cal
  A}(\varepsilon)\,|{r}(\lambda,\varepsilon,\vec{\omega})\kt\,.
 \label{defthem}
 \ee
Needless to add that what has to vanish are the auxiliary variables
$\Omega_{(a,b)}$ {\it alias} two functions which are available in
closed form. Thus, we have to solve the following two simplified
equations
 \be
 \sum_{k=1}^{N_{(\varrho)}}
 \,A^{(\varrho)}_{N_{(\varrho)},k}(\varepsilon)
 \,r_k^{(\varrho)}(\lambda,\varepsilon,{\omega}_{(a)},{\omega}_{(b)})=0\,,
 \ \ \ \ \ \ \varrho = a,b\,.
 \label{r0firovni}
 \ee
After the insertion of the respective matrix elements
$A^{(a,b)}_{N_{(a,b)},k}(\varepsilon)$ [cf. Eq.~(\ref{9})] we obtain
the pair of relations
 \be
  \left [
 \sum_{m=0}^{M-1} \,\epsilon^{m}\,\lambda\,
 V^{(M,M)}_{M-m,1}
 \right ]
 \omega_{(a)}+
  \left [
 \sum_{m=0}^{M-1} \,\epsilon^{m}\,\lambda\,
 V^{(M,N)}_{M-m,1}
 \right ]
 \omega_{(b)}
  =\epsilon^{M}\,\omega_{(a)}
 \,,
 \label{are1firovni}
  \ee
 \be
  \left [
 \sum_{n=0}^{N-1} \,\epsilon^{n}\,\lambda\,
 V^{(N,M)}_{N-n,1}
 \right ]
 \omega_{(a)}+
  \left [
 \sum_{n=0}^{N-1} \,\epsilon^{n}\,\lambda\,
 V^{(N,N)}_{N-n,1}
 \right ]
 \omega_{(b)}
  =\epsilon^{N}\,\omega_{(b)}
 \,.
 \label{bre1firovni}
  \ee
Obviously, this set can be read as a generalized eigenvalue problem
which determines generalized eigenvectors $\vec{\omega}$ at a
$K-$plet of eigenenergies $\epsilon$ which are all defined as roots
of the corresponding generalized secular determinant.

\subsection{Hierarchy of relevance and reduced approximations}

In the generic case one has to assume that the matrix elements of
the perturbation do not vanish and that $\lambda$ is small, i.e.,
that also the eigenvalues $\varepsilon$ remain small. This enables
one to omit all of the asymptotically subdominant corrections and to
consider just the linear algebraic system
 \be
 \left [
 \lambda\,
 V^{(M,M)}_{M,1}-\epsilon^{M}
 \right ]\,
 \omega_{(a)}+
 \lambda\,
 V^{(M,N)}_{M,1}\,
 \omega_{(b)}
  =0
 \,,
 \label{Aare1firovni}
  \ee
 \be
 \lambda\,
 V^{(N,M)}_{N,1}
 \omega_{(a)}+
  \left [
 \lambda\,
 V^{(N,N)}_{N,1}-\epsilon^{N}
 \right ]
 \omega_{(b)}
  =0
 \,.
 \label{Abre1firovni}
  \ee
The solution of these two simplified coupled linear relations exists
if and only if the determinant of the system vanishes,
 \be
 \det
 \left [
 \begin{array}{cc}
 \lambda\,V^{(M,M)}_{M,1}-\epsilon^{M}& \lambda\,
 V^{(M,N)}_{M,1}\\
 \lambda\,V^{(N,M)}_{N,1}& \lambda\,
 V^{(NN)}_{N,1}-\epsilon^{N}
 \ea
 \right ]=0\,.
 \label{secequ}
 \ee
Thus, an ordinary eigenvalue problem is encountered when $M=N$.

\begin{lemma}\label{lemma1}
In the generic equipartitioned cases with $M=N\geq 3$  the spectrum
ceases to be all real under bounded perturbations. The loss of
unitarity is encountered.
\end{lemma}
\begin{proof}
Both of the roots $\epsilon^N=\lambda\,x$ of the exactly solvable
quadratic algebraic secular equation (\ref{secequ}) may be
guaranteed to be real in a certain domain of parameters. Still, some
of the energies themselves are necessarily complex since
$\epsilon=\sqrt[N]{\lambda\,x}$ is an $N-$valued function with
values lying on a complex circle.
\end{proof}
In the other, non-equipartitioned dynamical scenarios with, say,
$M>N$, the behavior of the system in an immediate vicinity of its EP
extreme is still determined by the asymptotically dominant part of
the secular equation. After an appropriate modification the above
proof still applies.

\begin{lemma}\label{lemma2}
In the generic case with $M>N \geq 3$ we get, from the dominant part
of the generalized eigenvalue problem~(\ref{secequ}), a subset
($N-$plet) of asymptotically dominant eigenvalues $\varepsilon={\cal
O}({\lambda}^{1/N})$ which cannot be all real.
 \end{lemma}


\subsection{Unitary case: rescaled perturbations}

The loss of unitarity occurring in the $L=1$ models was associated
with the use of the too broad a class of norm-bounded perturbations.
From our preliminary results described in preceding subsection one
can conclude that a similar loss of unitarity may be also expected
to occur, in the generic case, at $L=2$. Indeed, the anisotropy of
the physical Hilbert space which reflects the influence of an
EP-related singularity of the Hamiltonian may be expected to lead
again to a selective enhancement of the weight of certain specific
matrix elements of perturbations $\lambda\,V^{(K)}$.

Once we recall the $L=1$ scenario of subsection \ref{unive} we
immediately imagine that the main source of the apparent
universality of the instability under norm-bounded perturbations
should be sought, paradoxically, in the routine but, in our case,
entirely inadequate norm-boundedness assumption itself. Indeed, in a
way documented by Lemmas \ref{lemma1} and \ref{lemma2}, the loss of
the reality of spectra may directly be attributed to the
conventional and comfortable but entirely random, unfounded and
formal assumption of the uniform boundedness of the matrix elements
of the perturbations.

Most easily the latter result may be illustrated using the
drastically simplified version (\ref{secequ}) of the leading-order
secular equation. Dominant role is played there by the quadruplet of
matrix elements $V^{(P,Q)}_{(P,1)}$ with superscripts $P$ and $Q$
equal to $M$ or $N$. Once they are assumed $\lambda-$independent and
non-vanishing, the leading-order energies read
$\epsilon=\sqrt[N]{\lambda\,x}$. Thus, at any $N \geq 3$ their
$N-$plet forms an equilateral $N-$angle in the complex plane of
$\lambda$.

The latter observation inspires a remedy. In a way eliminating the
$N \geq 3$ complex-circle obstruction one simply has to re-scale the
energies as well as all of the relevant matrix elements of the
perturbation. In this manner the ansatz
 \be
 \epsilon=\epsilon(E)=\sqrt{\lambda}\,E\,
 \ee
opens the possibility of the spectrum being real. Another multiplet
of postulates
 \be
 V^{(M,M)}_{M-m,1}
 =\lambda^{(M-m)/2}\,W^{(M,M)}_{M-m,1}
 \,,
 \ \ \ \
 V^{(M,N)}_{M-m,1}
 =\lambda^{(M-m)/2}\,W^{(M,N)}_{M-m,1}
 \,,
 \ \ \ \ m=0, 1, \ldots, M-1\,,
 \label{ameame}
  \ee
and
 \be
 V^{(N,M)}_{N-n,1}
 =\lambda^{(N-n)/2}\,W^{(N,M)}_{N-n,1}
 \,,
 \ \ \ \
 V^{(N,N)}_{N-n,1}
 =\lambda^{(N-n)/2}\,W^{(N,N)}_{N-n,1}
 \,,
 \ \ \ \ n=0, 1, \ldots, N-1\,
 \label{ameame2}
  \ee
now contains a new partitioned matrix $W=W^{(K)}$ which is assumed
uniformly bounded.

\begin{lemma}
There always exists a non-empty $(2M + 2N)-$dimensional domain
${\cal D}$ of the ``physical'' matrix elements of $W$ for which the
leading-order spectrum is all real and non-degenerate, i.e., in the
language of physics, tractable as stable bound-state energies.
\end{lemma}
\begin{proof}
The main consequence of the amended, necessary-condition
perturbation-smallness requirements (\ref{ameame}) and
(\ref{ameame2}) is that in our initial, unreduced leading-order
generalized eigenvalue problem (\ref{are1firovni}) +
(\ref{bre1firovni}), all terms in the sums become of the same order
of magnitude. By the scaling we managed to eliminate the explicit
presence of the measure of smallness $\lambda$. In other words the
input information about dynamics is formed now by the re-scaled
perturbation matrix $W$ which offers an $(2M + 2N)-$plet of free
${\cal O}(1)$ parameters. At the same time, the spectral-definition
output is given by the $M+N$ roots of the corresponding secular
equation, i.e., by the roots of an $(M+N)-$th-degree polynomial in
$E$, with all of its separate coefficients bounded and, in general,
non-vanishing. Under these conditions the assertion of the lemma is
obvious.
\end{proof}


\section{Discussion\label{seafive}}

\subsection{Schr\"{o}dinger picture and quasi-Hermitian Hamiltonians}

In the context of the new, Bender- and Boettcher-inspired paradigm
the non-Hermiticity of $H({\gamma})$ in ${\cal K}$ may be considered
compatible with unitarity whenever the spectrum itself is found
real. The explanation of the apparent paradox is easy: Under certain
reasonable mathematical assumptions (cf. \cite{Geyer}) one can find
an non-unitary invertible mapping $\Omega=\Omega({\gamma})$ with
property $\Omega^\dagger \Omega=\Theta \neq I$ which makes our
Hamiltonian self-adjoint,
 \be
 H({\gamma}) \to \mathfrak{h}({\gamma})=
 \Omega({\gamma})\,H({\gamma})\,\Omega({\gamma})
 =\mathfrak{h}^\dagger({\gamma})\,.
 \label{redun}
 \ee
In practice the idea enables one to avoid the use of the
``conventional'' Hamiltonian $\mathfrak{h}({\gamma})$ (self-adjoint,
by construction, in another Hilbert space ${\cal L}$) whenever it
happens to be ``prohibitively complicated''. In the literature such
a type of simplification of calculations is usually attributed to
Dyson \cite{Dyson}. Expectedly, the strategy (also known as
preconditioning) proved efficient as a tool of
construction of bound states in nuclear
physics \cite{Geyer}. In spite of a certain initial doubts
\cite{Jones,Jonesb}, the approach also proved applicable
in the quantum physics of scattering \cite{scatt,scattb}.

From a more abstract theoretical point of view the reference to the
``Dyson's'' isospectral mapping (\ref{redun}) becomes redundant when
we reclassify the space ${\cal K}$ as ``unphysical'', and when we
redefine its inner product yielding another, ``physical'' Hilbert
space $\cal H$,
 \be
 \br \psi_1|\psi_2\kt_{\cal H}=\br \psi_1|\Theta|\psi_2\kt_{\cal K}\,,
 \ \ \ \ |\psi_{1,2}\kt \in {\cal K}
 \,.
 \label{relat}
 \ee
A new, equivalent, ``two-Hilbert-space'' version of the
Schr\"{o}dinger picture is obtained in which the physics described
in the correct Hilbert space $\cal H$ is ``translated'' to its
mathematically easier representation in  $\cal K$. In this sense,
the self-adjointness of hypothetical $\mathfrak{h}$ in hypothetical
${\cal L}$ is found equivalent to the self-adjointness of $H$ in
${\cal H}$, represented by the relation
 \be
 H^\dagger({\gamma})\,\Theta({\gamma})=\Theta({\gamma})\,H({\gamma})\,.
 \label{quasihe}
 \ee
Dieudonn\'e \cite{Dieudonne} and Scholtz et al \cite{Geyer}
suggested to call relation (\ref{quasihe}) the quasi-Hermiticity of
$H$ in ${\cal K}$. In the context of quantum phenomenology a
decisive amendment of the formalism may be seen in the split of the
description of dynamics with information carried by the two
operators $H({\gamma})$ and $\Theta({\gamma})$ defined in ${\cal K}$
in place of one (viz., of $\mathfrak{h}({\gamma})$ living in ${\cal
L}$).

\subsection{Non-Hermitian degeneracies with $L>2$}

After the present decisive clarification of the possibility of the
replacement of the $L=1$ formalism by its $L=2$ generalization, the
next move to the further, $L>2$ dynamical scenarios is now an
almost elementary exercise. Indeed, in sections \ref{sectionfour}
and \ref{sectionfive} it would be sufficient to move from the $L=2$
partitioning of $K={N_{(a)}+ N_{(b)}}$ to its arbitrary $L>2$
analogues (\ref{papas}) and to the related partitioned vector sets
[sampled, e.g., by Eq.~(\ref{Mde})] and matrices [sampled, e.g., by
Eq.~(\ref{geneve})]. The $L=2$ doublets of the eligible superscripts
$^{(a)}$ and $^{(b)}$ marking the curly kets [cf. (\ref{rajt}) or
(\ref{r0firovni}), etc] may be very easily extended to the $L-$plets
with $L>2$, etc.

Along these lines the form of our basic power-series expansion of
the perturbed bound-state kets (\ref{[36]}) remains unchanged. The
related $L=2$ compatibility constraint (\ref{firovni}) must only be
replaced by an $L-$plet of equations
 \be
 y_{N_{(a_{j})}}^{(solution)}(\lambda,\varepsilon,\vec{\omega})=0\,,
 \ \ \ \ \ j=1,2,\ldots,L
 \,
 \label{asoco}
 \ee
where the unknown vector $\vec{\omega}$ has $L$ components.

\subsection{The next-to-leading-order approximation}

In our present paper we did not pay too much attention to the
next-to-leading-order (NLO) approximation where one would have to
set
 \be\,
 |{\vec{y}}^{(solution)}(\lambda,\varepsilon,\vec{\omega})\kt=
  {\cal A}(\varepsilon)\,|{r}(\lambda,\varepsilon,\vec{\omega})\kt
 - \lambda\,{\cal A}(\varepsilon)\,V\,
  \Pi\,{\cal
  A}(\varepsilon)\,|{r(\lambda,\varepsilon,\vec{\omega})}\kt\,.
  \ee
Our main reason was that such a generalization would make the
associated compatibility conditions (\ref{asoco}) perceivably more
complicated. Indeed, in contrast to the leading-order case, an
almost inessential improvement of the insight in the qualitative
features of the quantum system in question would be accompanied, in
the NLO formulae, by the emergence of multiple new matrix elements
of perturbations $\lambda\,V^{(K)}$ including even the terms which
would be quadratic functions of these matrix elements.

This being said, it is necessary to add that even on the pragmatic
and qualitative level, the omission of the NLO corrections only
seems completely harmful in the open-system setting using bounded
matrices $\lambda\,V^{(K)}$ where one does not insist on having the
strictly real spectrum. The point is that whenever one has to
guarantee the unitarity of the system, the class of the admissible
perturbations must be further restricted. In this sense,
unfortunately, an analysis using NLO might prove necessary. Indeed,
the use and precision of the leading-order approximation need not be
sufficiently reliable in general. Even a quick glimpse at the
underlying assumptions (\ref{ameame}) and (\ref{ameame2}) reveals
that the leading-order approximation does not incorporate the
influence of a large subset of the matrix elements of
$\lambda\,V^{(K)}$. At the same time, one has to keep in mind that
in our present approach the dimension of the matrices $K$ has been
assumed finite. For this reason, in the case of doubts, a turn to
the more universal and brute-force numerical methods might prove to
be, in practical calculations, a reasonable alternative to the
rather lengthy and complicated NLO calculations.

\subsection*{Acknowledgements}

This work is
supported by the Faculty of Science of the University of Hradec Kr\'{a}lov\'{e}.
The author also acknowledges the financial support
from the
Excellence project P\v{r}F UHK~2020 Nr. 2212.

\newpage

\section*{\label{Appendix}Appendix A.
An exhaustive classification of the degeneracies of exceptional points}

Up to an arbitrary permutation, every partitioning
 \be
 K=M_1+M_2+\ldots+M_L
 \label{papas}
 \ee
of the full matrix dimension characterizes a different system and
its NHD limit. Thus, we have to postulate, say,
$$M_1 \geq M_2 \geq \ldots \geq M_L \geq 2$$ and introduce the
schematic multi-indices \fbox{$M_1+M_2+\ldots+M_L$} in a way
illustrated in Table \ref{pexp4}.

\begin{table}[h]
\caption{Eligible EP-related partitionings}\label{pexp4}

\vspace{2mm}

\centering
\begin{tabular}{||c||c||}
\hline \hline
   $K$  & list
    \\
 \hline \hline
 2& \fbox{2}\\
 3& \fbox{3}\\
 4& \fbox{4}\ \fbox{2+2} \\
 5&\fbox{5}\ \fbox{3+2} \\
 6  &\fbox{6}\ \fbox{4+2}\
 \fbox{3+3}\
 \fbox{2+2+2} \\
 7&\fbox{7}\
 \fbox{5+2}\
 \fbox{4+3}\
 \fbox{3+2+2}\\
 8&\fbox{8}\
 \fbox{6+2}\
 \fbox{5+3}\
 \fbox{4+4}\
 \fbox{4+2+2}\
 \fbox{3+3+2}\
 \fbox{2+2+2+2}\\
 9&\fbox{9}\
 \fbox{7+2}\
 \fbox{6+3}\
 \fbox{5+4}\
 \fbox{5+2+2}\
 \fbox{4+3+2}\
 \fbox{3+3+3}\
 \fbox{3+2+2+2}\\
  10 &\fbox{10}\
 \fbox{8+2}\
 \fbox{7+3}\
 \fbox{6+4}\
 \fbox{6+2+2}\
 \fbox{5+5}\
  \fbox{5+3+2}\
 \\
   &
 \fbox{4+4+2}\
 \fbox{4+3+3}\
 \fbox{4+2+2+2}\
 \fbox{3+3+2+2}\
 \fbox{2+2+2+2+2}\\
 \vdots&   \dots\\
 \hline
 \hline
\end{tabular}
\end{table}

The Table indicates that the
nontrivial $L>1$ partitionings not containing trivial items $M_j=1$
only exist at $K \geq 4$. The value of the count ${\cal N}(K)$ of
the separate nontrivial, $L\geq 2$ items is seen to exceed one only
at $K=6$. Nevertheless, this count starts growing quickly at the
larger $K$s (see Table~\ref{pexp5}). Thus, the current practice of
studying just the EPK models with $L=1$ misses in fact the huge
majority of the alternative NHD scenarios. Marginally, it is worth
adding that the counts ${\cal N}(K)$ form a well known sequence. In
the open-access on-line encyclopedia of integer sequences
\cite{A083751} it is assigned the code number A083751.

\begin{table}[h]
\caption{The sequence of counts ${\cal N}(K)$ with $K=2,3,\ldots$.}
\label{pexp5}

\vspace{2mm}

\centering
\begin{tabular}{||l||}
 \hline \hline
 0, 0, 1, 1, 3, 3, 6, 7, 11, 13, 20, 23, 33, 40, 54, 65, 87, 104, 136, 164, \\
 209, 252, 319, 382, 477, 573, 707, 846, 1038, 1237, 1506, 1793, \ldots \, .\\
 \hline
 \hline
\end{tabular}
\end{table}

 \noindent

For our present purposes, the asymptotic growth of sequence ${\cal
N}(K)$ (which is slightly slower than exponential) as well as its
precise mathematical definition are less relevant. At the realistic,
not too large dimensions $K$, for example, we might also use some
alternative definitions. One of them leads to values
 $
{\cal N}(K) 
 $
equal to the first differences (diminished by one) of the
special
partitions (of the code number $A000041$, cf. \cite{A000041}),
%
or to the first differences of the numbers of trees of diameter four
(see the integer sequence with code number $A000094$ in \cite{A000094}).
%
%
%
%
%
Nevertheless, irrespectively of the choice of definition let us
point out that in the NHD vicinity, the partitioning multi-indices
will classify the phenomenologically non-equivalent physical systems
in general. As long as the superscripts $^{(K)}$ are in fact
redundant, we will omit or replace them by the more relevant
information about the partitioning, therefore. In particular, symbol
 \be
  {\cal J}^{( M_1\bigoplus M_2\bigoplus \ldots \bigoplus M_L)}(\eta)
  \ee
will represent the general block-diagonal canonical representation $
{\cal J}^{(K)}(\eta)$ of the Hamiltonian defined in
Eq.~(\ref{geneve}).

\newpage


\section*{Appendix B. Perturbation theory near non-degenerate exceptional points}

\subsection*{B.1. The choice of basis at $L=1$}

Relation (\ref{fealt}) can be read as a definition of transition
matrix responsible for the canonical $K$ by $K$ Jordan-block
representation
 \be
 {J}^{(K)}(\eta)=
 \left [Q^{(EPK)}\right ]^{-1}\,
 H(\gamma^{(EPK)})\, Q^{(EPK)}
 \,
  \label{refealt}
 \ee
of the EPK $L=1$ limit of any given non-Hermitian but ${\cal
PT}-$symmetric Hamiltonian of phenomenological relevance. From this
point of view one can extend the same transformation (i.e., for
matrices with $K < \infty$, a mere choice of the basis in Hilbert
space) to a vicinity of the EPK singularity. This yields the Jordan
block matrix plus perturbation,
 \be
 \left [Q^{(EPK)}\right ]^{-1}\,
 H^{(K)}(\gamma)\, Q^{(EPK)}=
 {J}^{(K)}(\eta) +  \lambda\,V^{(K)}(\gamma)\,.
  \label{berefealt}
 \ee
Such a definition contains a redundant but convenient measure
$\lambda = {\cal O}(\gamma-\gamma^{(EPK)})$ of the smallness of
perturbation.

Due to the conventional postulate of having a specific
one-parametric family of Hamiltonians $H^{(K)}(\gamma)$ given in
advance, the introduction of the concept of perturbation in
(\ref{berefealt}) is just a formal step. Nevertheless, the
interaction term itself could be also reinterpreted as a
model-independent random perturbation which carries the input
dynamical information. From such a perspective every preselected
perturbation term defines a different Hamiltonian
$\widetilde{H^{(K)}})$ which merely coincides with $H^{(K)}(\gamma)$
in the NHD limit $\gamma \to \gamma^{(EP)}$ (a few exactly solvable
samples of such a truly remarkable Hamiltonian matching may be found
in \cite{BHAO}). This means that via relation
 \be
 {J}^{(K)}(\eta) +  \lambda\,V^{(K)}=
 \left [Q^{(EPK)}\right ]^{-1}\,
 \widetilde{H^{(K)}})\, Q^{(EPK)}\,
  \label{uberefealt}
 \ee
[i.e., via a mere reordering of relation (\ref{berefealt})] we
obtain a new picture of physics in which the resulting tilded
Hamiltonian with real spectrum need not be ${\cal PT}-$symmetric at
all.

\subsection*{B.2. The description of the
unfolding of the degeneracy at $L=1$}

A specific constructive extension of the latter observation has been
presented in our recent paper \cite{corridors}. We were able there
to prove that in the NHD dynamical regime, the mere boundedness of
the norm of matrix $V^{(K)}$ together with the smallness of
parameter $\lambda$ still {\em do not \,} guarantee the survival of
the unitarity of the system after perturbation. We showed there (cf.
also \cite{Ruzicka}) that one can guarantee the absence of a
``quantum catastrophe'' (i.e., of an abrupt change of some of the
system's observable features, see \cite{catast,catastb,catastc}) only via a certain
self-consistent revision of the criteria of smallness of matrix
$V^{(K)}$.

Having in mind the parameter-independence and invertibility of the
transformation matrix $Q^{(EPK)}$ the quantification of the
influence of the perturbation is of enormous interest, among others,
in the analysis of stability of the system in question \cite{Viola}.
This was the reason why we also addressed the $L=1$ perturbative
bound-state problem
    \be
 \left [
 {J}^{(K)}(0) + \lambda\,V^{(K)}
  \right ]
 \,|{\Psi}\kt=\epsilon\,|{\Psi}\kt\,,
  \label{perL1}
 \ee
with $\eta=0$ in Ref.~\cite{admissible}. With the ket-vector
subscripts $j$ in  $|{\Psi}_j\kt$ running from $1$ to $K$ we fixed
the norm (by setting $|{\Psi}_1\kt=1$), and we relocated the first
column of Eq.~(\ref{perL1}), viz., vector
 \be
 \vec{r} = \vec{r}(\lambda) = \left (
 \ba
 \epsilon - \lambda\,{V}_{1,1}\\
 - \lambda\,{V}_{2,1}\\
 \vdots\\
 - \lambda\,{V}_{K,1}
 \ea
 \right )\,
 \label{postaru}
 \ee
to the right-hand side of the equation. Then we restored the
comfortable square-matrix form of the equation via its two further
equivalent modifications. Firstly we added a new, temporarily
undetermined auxiliary component $\Omega_K$ to an ``upgrade'' of the
wave function
 \be
 \vec{y} =\left (
 \ba
 y_1\\
 y_2\\
 \vdots\\
 y_{K-1}\\
 y_{K}
 \ea
 \right )= \left (
 \ba
 |{\Psi}_2\kt\\
 |{\Psi}_3\kt\\
 \vdots\\
 |{\Psi}_{K}\kt\\
 \Omega_{K}
 \ea
 \right )\,.
 \label{jtarov}
 \ee
Subsequently, an introduction of the following auxiliary
lower-triangular $K$ by $K$ matrix
  \be
 A=A(K,\epsilon)=\left (
 \begin{array}{ccccc}
 1&0&0&\ldots&0\\
 \epsilon&1&0&\ddots&\vdots\\
 \epsilon^2&\epsilon&\ddots&\ddots&0\\
 \vdots&\ddots&\ddots&1&0\\
 \epsilon^{K-1}&\ldots&\epsilon^2&\epsilon&1
 \ea
 \right )\,
 \label{9}
 \ee
and of its two-diagonal inverse
 \be
 A^{-1}
 =\left (
 \begin{array}{ccccc}
 1&0&0&\ldots&0\\
 -\epsilon&1&0&\ddots&\vdots\\
 0&-\epsilon&\ddots&\ddots&0\\
 \vdots&\ddots&\ddots&1&0\\
 0&\ldots&0&-\epsilon&1
 \ea
 \right )\,
 \label{9inv}
  \ee
accompanied by a parallel formal upgrade of the interaction matrix,
 \be
 V^{(K)} \ \to \ Z
 =\left (
 \begin{array}{ccccc}
 {V}_{1,2}&{V}_{1,3}&\ldots&{V}_{1,K}&0\\
 {V}_{2,2}&{V}_{2,3}&\ldots&{V}_{2,K}&0\\
 \ldots&\ldots&\ldots&\vdots&\vdots\\
 {V}_{K,2}&{V}_{K,3}&\ldots&{V}_{K,K}&0
 \ea
 \right )\,
 \label{ustahromu}
 \ee
enabled us to rewrite our initial homogeneous Schr\"{o}dinger
Eq.~(\ref{perL1}), in the last step of the construction of its
solution, in an equivalent matrix-inversion form
 \be
 (A^{-1} + \lambda\,Z) \vec{y}= \vec{r}\,
 \label{tarov}
 \ee
or, better,
 \be
 (I + \lambda\,A\,Z) \vec{y}= A\,\vec{r}\,
 \label{satarov}
 \ee
accompanied by the innocent-looking but important self-consistence
constraint
 \be
 \Omega_K=0\,.
 \label{compat}
 \ee
In the leading-order application of the recipe we then returned to
the slightly vague assumption of the ``sufficient smallness'' of the
perturbation. On these grounds we recalled the formal Taylor-series
expansion of the resolvent which yielded the closed formula
 \be
 \vec{y}^{(solution)}(\epsilon)= A\,\vec{r}-
 \lambda \,A\,\,Z\,
 A\,\vec{r}+\lambda^2 \,A\,\,Z\,A\,Z\,
 A\,\vec{r}
 -\ldots\,
 \label{tadef}
 \ee
for the modified wave function. It contained a free parameter
$\epsilon$ which had to be fixed via the supplementary secular
equation (\ref{compat}). In the light of the Taylor-series formula
(\ref{tadef}), such a secular equation now acquires the
$K-$th-vector-component form
 \be
 {y_K}^{(solution)}(\epsilon)=0\,.
 \label{krutadef}
 \ee
of an explicit transcendental equation for the energies $\epsilon$.

\subsection*{B.3. Unitary-evolution process of unfolding at $L=1$\label{unive}}

The latter constraint (\ref{compat}) plays the role of an implicit
definition of the spectrum. The $K-$plet of roots $\epsilon_n
=\epsilon_n(\lambda)$, $n=1,2,\ldots,K$ represents the bound-state
energies. After the truncation of the series, just approximate
solutions are being obtained. In Ref.~\cite{admissible},
incidentally, even the leading-order roots were found complex in
general.

This observation was interpreted as indicating that in an immediate
EPK vicinity the norm-bounded perturbations $\lambda\,{V}$ should
still be considered, in the unitary theory, ``inadmissibly large''.
The non-unitary, open-quantum system interpretation of the
perturbations proved needed forcing the system to perform, at an
arbitrarily small but non-vanishing $\lambda\neq 0$, an abrupt
quantum phase transition.

Incidentally, qualitatively the same conclusions were also obtained
in the above-mentioned more concrete study \cite{Uwe} of the
specific Bose-Hubbard model in its EPK dynamical regime. The
resolution of an apparent universal-instability paradox was provided
in our subsequent study \cite{corridors} in which we studied the
underlying exact as well as approximate secular equations in more
detail. Our ultimate conclusion was that the necessary smallness
condition specifying the class of the admissible, unitarity
non-violating perturbations does not involve their upper-triangular
matrix part at all. In contrast, their lower-triangular matrix part
must be given the following, matrix-element-dependent form
 \be
 \lambda\,V^{(K)}_{(admissible)}=\left[ \begin {array}{cccccc}
  {\lambda}^{1/2}{{\mu}}_{11}
  &0&\ldots&0&0&0
  \\\noalign{\medskip}\lambda\,{\mu}_{{21}}&{\lambda}^{1/2}{{\mu}}_{22}
&\ldots&0&0&0
  \\\noalign{\medskip}{\lambda}^{3/2}
  \,{\mu}_{{31}}&\lambda\,{\mu}_{{32}}&\ddots&\vdots&\vdots&0
  \\\noalign{\medskip}{\lambda}^{2}{\mu}_{{41}}&{\lambda}^{3/2}
  \,{\mu}_{{42}}
  &\ddots&\ddots&0&0
  \\\noalign{\medskip}\vdots&\vdots&\ddots&\lambda\,{\mu}_{{{K}-1{K}-2}}&
  {\lambda}^{1/2}{{\mu}}_{{K}-1{K}-1}
 &0
  \\\noalign{\medskip}{\lambda}^{{K}/2}{{\mu}}_{{K}1}&
 {\lambda}^{({K}-1)/2}{\mu}_{{{K}2}}&\ldots&\lambda^{3/2}
 \,{\mu}_{{{K}{K}-2}}&\lambda\,{\mu}_{{{K}{K}-1}}&{\lambda}^{1/2}{{\mu}}_{{K}{K}}
 \end {array} \right]\,.
 \label{uho}
 \ee
During the decrease of $\lambda\to 0$, all of the variable
lower-triangle matrix-element parameters must remain bounded,
$\mu_{j,k}={\cal O}(1)$. In other words, as long as we are working
in a specific, fixed ``unperturbed'' basis, the matrix structure
(\ref{uho}) may be interpreted as manifesting a characteristic
anisotropy and the hierarchically ordered weights of influence of
the separate matrix elements. Indeed, we may rescale
 \be
 \lambda\,V^{(K)}_{(admissible)}=\lambda^{1/2}\,B(\lambda)\,
 V^{(reduced)}\,B^{-1}(\lambda)
 \label{216}
 \ee
where $B(\lambda)$ would be a diagonal matrix with elements
$B_{jj}(\lambda)=\lambda^{j/2}$ and where the reduced perturbation
matrix would be bounded, $V^{(reduced)}_{jk}={\cal O}(1)$.

On this necessary-condition background valid at all dimensions $K$,
the samples of sufficient conditions retain a purely numerical
trial-and-error character, with the small$-K$ non-numerical
exceptions discussed, in \cite{corridors}, for the matrix dimensions
up to $K=5$.

\end{document}